\documentclass[11pt, oneside]{article}   	

\usepackage{fullpage,times}
\usepackage[english]{babel}
\usepackage[toc,page]{appendix}
\usepackage[margin=1in]{geometry}
\usepackage{hyperref}
\usepackage{graphicx,wrapfig,lipsum,xcolor}	
\usepackage[labelformat=empty]{caption}
\usepackage{yfonts}	
\usepackage{amssymb}
\usepackage{amsmath}
\usepackage{amsthm}
\usepackage{scalerel}
\usepackage{verbatim}
\usepackage[lined,boxed,ruled,norelsize,algo2e,linesnumbered]{algorithm2e}

\newcommand{\E}{\mathbb{E}}

\usepackage[colorinlistoftodos]{todonotes}

\newtheorem{thm}{Theorem}[section]
\newtheorem{lemma}[thm]{Lemma}
\newtheorem{prop}[thm]{Proposition}
\newtheorem{cor}[thm]{Corollary}

\theoremstyle{definition}
\newtheorem{defn}{Definition}

\newcommand{\NN}{{\mathbb{N}}}
\newcommand{\RR}{{\mathbb{R}}}

\newcommand{\argmin}{{\mathrm{argmin}}}
\newcommand{\match}{{\mathrm{match}}}
\newcommand{\serve}{{\mathrm{serve}}}

\DeclareMathOperator{\poly}{poly}

\title{Online Multiserver Convex Chasing and Optimization}

\author{S\'ebastien Bubeck\thanks{Microsoft Research, {\tt sebubeck@microsoft.com}} \and
            Yuval Rabani\thanks{The Rachel and Selim Benin School of Computer Science and Engineering, The Hebrew University of Jerusalem, Jerusalem 9190416, Israel, {\tt yrabani@cs.huji.ac.il}. Research supported in part by ISF grant 956-15 and by NSFC-ISF grant 2553-17. Part of this work was done while visiting Microsoft Research.}\and 
            Mark Sellke\thanks{Stanford University Department of Mathematics, {\tt msellke@stanford.edu}. Part of this work was done while visiting Microsoft Research. Partially supported by NSF and Stanford graduate fellowships.}
           }

\date{\today}

\begin{document}

\maketitle

\begin{abstract}
We introduce the problem of $k$-chasing of convex functions, 
a simultaneous generalization of both the famous $k$-server problem in $\RR^d$, and of the problem of chasing convex bodies and functions.
Aside from fundamental
interest in this general form, it has natural applications to online $k$-clustering
problems with objectives such as $k$-median or $k$-means. We show that this
problem exhibits a rich landscape of behavior. In general, if both $k > 1$ and 
$d > 1$ there does not exist any online algorithm with bounded competitiveness.
By contrast, we exhibit a class of nicely behaved functions (which include in particular
the above-mentioned clustering problems), for which we show that
competitive online algorithms exist, and moreover with dimension-free competitive ratio.

We also introduce a parallel question of top-$k$ action regret minimization in the realm
of online convex optimization. 
There, too, a much rougher landscape emerges for $k>1$.
While it is possible to achieve vanishing regret, unlike the top-one action case the
rate of vanishing does not speed up for strongly convex functions. Moreover,
vanishing regret necessitates both intractable computations and randomness.
Finally we leave open whether almost dimension-free regret is achievable for $k>1$ and
general convex losses. As evidence that it might be possible, we prove dimension-free regret for linear 
losses via an information-theoretic argument.
\end{abstract}

\thispagestyle{empty}
\newpage
\setcounter{page}{1}


\section{Introduction}

\paragraph{Motivation and problem statement.}
The $k$-server problem~\cite{MMS90} is a fundamental question in online computing. In this
problem, $k$ servers occupy $k$ points in a metric space, and move around to serve on-the-fly
a sequence of incoming requests. Each request is a point that must be served by moving at
least one server to its location, incurring a cost equal to the distance traveled. In the competitive
analysis paradigm of online computing, the goal of an online algorithm is to serve any request
sequence at a cost proportional to the best possible cost for that sequence (allowing in general
also a fixed additive term). The famous $k$-server conjecture states that the factor of proportionality
is precisely $k$ in every metric space of size at least $k+1$. Work on this conjecture and on the
counterpart randomized $k$-server conjecture (see~\cite{KRR91}) has for a long time defined 
and dominated the area of online computing, and also produced many of its deepest revelations.
The allure of the $k$-server model also stems from the fact that it is a fairly general model of
online computing that captures as special cases a wide range of applications. The $k$-server 
conjecture itself is nearly settled~\cite{MMS90,KP95}, and significant progress has recently been 
reported on its randomized counterpart~\cite{BLMN03,BBMN11,BCLLM18,Lee18}.

In this paper we discuss the problem of $k$-chasing of convex functions, which is a generalization
of the $k$-server problem in $\RR^d$ that is natural, revealing, and has interesting applications 
to problems of decision making under uncertainty and unsupervised learning. The only difference with
the classical $k$-server problem is that in this problem the requests are convex penalty functions.
An algorithm pays the sum of the movement cost and the penalty cost. (Penalties can, of course, 
be levied on the algorithm either before or after the move, leading to two distinct flavors of the 
problem; we discuss this matter further below.) Chasing convex functions online (by a single server) 
is also a problem with a long history, starting with the convex body chasing problem of~\cite{FL93}
(a chased convex body can be viewed as a convex function that is $0$ in the body and $+\infty$
outside the body). Major breakthroughs were reported recently 
in~\cite{BBEKU18,BLLS19,ABCGL19,BKLLS20,AGGT20,Sel20}. 
However, once $k > 1$ servers are chasing, convexity (viewing the problem as one server chasing
in $(\RR^d)^k$) is lost. So, one aspect of our research is the study of non-convex chasing when 
there is still some underlying convex structure. Another aspect is that chasing with $k$ servers 
models online convex optimization in a distributed setting. 

A more concrete motivation is that a 
special case of this problem, not captured by the $k$-server setting, is dynamic $k$-clustering. 
Consider, for instance, $k$-median or $k$-means clustering of points in $\RR^d$ that arrive 
sequentially online. The partition into clusters is determined by the locations of the $k$ centers
of the clusters. When a point arrives, we pay the cost of adding it to the cluster with nearest
center (either before or after adjusting the centers and paying the cost of this movement). 
We can think of the current centers as a model 
predicting future input based on the pattern shown by past input. The cost function reflects the 
tradeoff between the quality of the prediction and the stability of the clustering. To the best of
our knowledge, this framework for online clustering has not been addressed previously. Some 
alternative notions of online clustering were inverstigated in~\cite{CCFM97,LSS16,CGKR19,BR20}.

The classical $k$-server problem can be viewed as a special case of this problem where each 
convex function is $0$ at the requested point and $+\infty$ elsewhere. This case demonstrates 
that it does not make sense to charge a penalty prior to movement (a setting which we call
{\em blind chasing}) without making some restrictive 
assumptions on the requested functions, for instance Lipschitz continuity. If the requested functions
are Lipschitz continuous, then it also makes sense to study this setting in the regret minimization
paradigm of online convex optimization~\cite{Sha12}. In this setting, the online algorithm chooses 
in each step an action to play---a point in the unit ball in $\RR^d$. Then it gets penalized for its
chosen action. In each step, the penalties for all possible actions form an $L$-Lipschitz convex 
function. Switching actions from step to step is not charged. The goal is to minimize the regret, 
an additive guarantee against the best static action, rather than a multiplicative guarantee against 
the best dynamic solution. In our setting, the online algorithm can choose simultaneously 
$k$ actions to play, it pays the minimum penalty for the $k$ choices, and it seeks to minimize
the regret against the best static choice of $k$ actions. In machine learning
this metric (best of $k$ actions) corresponds to the often reported top-$k$ accuracy (i.e., an
algorithm can make $k$ predictions, and is evaluated on the most accurate one of them).
The goal in designing algorithms is to have
the average regret per step vanishing quickly as the number of steps grows, or in other words
that the total regret grows at a rate far below linear.

Thus, all these problems are characterized by the dimension $d$, the number of servers $k$, and 
the Lipschitz continuity constant $L$, and the main goal is to bound as tightly as possible the 
competitive ratio or the regret as a function of these parameters.

\paragraph{Our results.}
We give a deterministic $O(k)$-competitive online algorithm for chasing convex functions on $\RR$
using $k$ servers. A $d$-competitive algorithm was known from previous work~\cite{AGGT20,Sel20}
for chasing convex functions in $\RR^d$ using one server. We show in contrast to our result and these
previous results that if both $k,d > 1$, there does not exist any algorithm achieving bounded
competitiveness, even for the more restricted problem of chasing convex bodies. (Chasing convex
bodies is equivalent to chasing convex functions for one server~\cite{BLLS19}, but might be easier for
multiple servers.) Despite this negative result, we show that nonetheless for the convex functions
that arise in online $k$-median and $k$-means clustering, and more generally for well-sharpened
convex functions (defined in Section~\ref{sec:wellsharp}), there exist competitive online algorithms. 
In particular, for the online $k$-median problem the adaptive online competitive ratio is $O(k)$ and 
the randomized (oblivious) competitive ratio is $\poly\log k$. For the online $k$-means problem the 
adaptive online competitive ratio is $O(k^2)$. In fact, these results apply to any metric space and
not just to Euclidean space. The meaning of {\em adaptive online} and {\em oblivious} 
follows~\cite{BBKTW90}. For the sake of brevity, it suffices to say that there isn't a 
significant difference between the adaptive online competitive ratio and the deterministic competitive 
ratio.

We also investigate regret guarantees for online convex optimization. Here we shall assume that 
the actions are in the unit ball in $\RR^d$, and the penalty functions are $1$-Lipschitz and convex. 
All the results scale linearly with the Lipschitz constant. It was known that gradient descent achieves 
a regret of $O(\sqrt{T})$ after $T$ steps~\cite{Zin03}. If the penalty functions are $\alpha$-strongly 
convex, then the upper bound drops to $O(\alpha^{-1} \log T)$~\cite{HAK07}. In the 
$k$-action setting, there is trivially a randomized online algorithm that achieves regret $O(\sqrt{dkT\log T})$,
simply by discretizing the unit ball and using the scheme of prediction from expert advice~\cite{CFHHSW97}.
This observation was already made in~\cite{CGKR19} for the special case of the $k$-means objective.
Notice that such a strategy requires time exponential in $dk$. We show
that three main aspects of the upper bound are unavoidable, implying a price that needs to be
paid for the loss of convexity when $k > 1$. Firstly, we show that unless P$=$NP, the 
super-polynomial dependence on $d$ is required to get regret below $T/d$, even for $k=2$.
Secondly, we show that in contrast to the $k=1$ case where randomization does not help, 
for $k > 1$ no deterministic online algorithm can achieve sublinear regret. Thirdly, we show that 
even for strongly convex functions, randomized algorithms cannot achieve regret better than 
$\Omega\left(k^{-2/(d-1)}\cdot\sqrt{kT}\right)$ if $k,d > 1$, as opposed to the logarithmic regret 
if $k=1$. These results leave open the question of whether there exists a method to achieve sublinear 
regret that is independent of the dimension $d$. We resolve this question for the special case 
of linear penalty functions. We show that in this case, regret of $\tilde{O}(\sqrt{kT})$ is achievable 
in any dimension.
The proof of this latter result is information-theoretic, and does not give an explicit algorithm, though one can always
solve the minimax problem to obtain an algorithm that runs in time exponential in $T$.



\section{Chasing Convex Functions with ${\mathbf k}$ Servers}

Consider the set ${\cal F}$ of positive convex functions 
$f:\RR^d\rightarrow (0,+\infty]$.
In the $k$-chasing of convex functions problem, $k$ servers are placed 
at some initial configuration in $\RR^d$, and must serve a sequence of
requests from ${\cal F}$. In a configuration 
$\{x_1,x_2,\dots,x_k\}\in{\RR^d \choose k}$, the cost of serving a request
$f\in{\cal F}$ is $\min_{1 \leq i \leq k} f(x_i)$. Upon receiving a request $f$, one
or more servers can move to improve the service cost, but in doing so
they incur a movement cost equal to the total distance traveled. The
distance is measured according to some metric on $\RR^d$, which we
generally assume to be the Euclidean ($L^2$) metric. Our goal is
to investigate the design of competitive online algorithms for this problem,
minimizing the worst case ratio between the total service and movement
cost of the algorithm and the same measurement for an optimal prescient
solution.




\subsection{The line case}

We present here a deterministic $O(k)$-competitive algorithm for $k$-chasing of convex functions
in $\RR$ (i.e., $d=1$) with any $k \geq 1$.

For $k=1$ the following simple strategy from \cite{bansal20152}
is competitive:
Upon receiving a request $f\in {\cal F}$ when the server is in position $x\in\RR$, move
towards $\argmin f$ while the service cost at the current location exceeds the movement
cost in the current step. Notice that the condition for terminating the move may never be
reached, and in this case the server serves $f$ at $\argmin f$.

Our proposed algorithm for general $k$ is a double coverage (a la~\cite{CKPV91}) generalization of 
the simple $k=1$ algorithm. It is defined as follows: Naturally, if a server is located at $\argmin f$, do nothing. 
Otherwise, if $\argmin f$ falls outside the minimal interval containing the $k$ servers, move 
the nearest server (one of the two extremes) towards $\argmin f$ while the service cost in its
current location exceeds its movement cost in the current step. Otherwise, $\argmin f$ is located 
between the positions of two adjacent servers. Move both of them towards $\argmin f$ at the same 
rate, while the service cost (incurred at the current location of one of the two moving servers) exceeds
the movement cost in the current step and $\argmin f$ is not yet reached by a server.

\begin{thm}\label{thm: one-dim competitive}
The above algorithm is $4k$-competitive for $k$-chasing convex functions in $\RR$.
\end{thm}

The proof of this result combines the classical potential function of double coverage on a line with the analysis of the simple $k=1$ algorithm for convex function chasing. We defer it to Appendix \ref{app:A}.

\subsection{Two servers in the plane is already hard}

We showed that if either $k=1$ or $d=1$, then there exist competitive algorithms for $k$-chasing convex functions in $\RR^d$. We now prove that these are in fact the {\em only} cases for which one can be competitive:

\begin{thm} \label{thm:lb0}

When $k,d\geq 2$, there do not exist competitive algorithms for $k$-chasing convex bodies in $\RR^d$. This holds even for randomized algorithms against an oblivious adversary.

\end{thm}
Note that the above impossibility already holds for chasing convex bodies, which is a special case of chasing convex functions.

To prove this result we first show, with a classical argument, that for a single server one cannot chase union of two intervals on the line, see Lemma \ref{lem:MTSlb}. Then the proof of Theorem \ref{thm:lb0} proceeds by showing that one can simulate the setting of Lemma \ref{lem:MTSlb} with a certain sequence of convex sets to be chased in the plane with two servers.

\begin{lemma}

\label{lem:MTSlb} 

Chasing unions of two intervals $[a_1,a_2]\cup [b_1,b_2]\subseteq [0,1]$ with a single server has infinite competitive ratio, even with randomization against the oblivious adversary. 

\end{lemma}

\begin{proof}
The lower bound uses the following well-known result (for instance~\cite{KRR91}). There exists a monotonically increasing function
$f:\NN\rightarrow\NN$ for which $f(n)\rightarrow\infty$ as $n\rightarrow\infty$ such that the following holds. For every $n\in\NN$,
$n\ge 2$, every randomized online algorithm for the $(n-1)$-server problem on the point set $\{1,2,\dots,n\}$ endowed with the
distance function $d(i,j) = |i-j|$ incurs a competitive ratio exceeding $f(n)$ against the oblivious adversary.

The lemma follows from a simple reduction of the above $k$-server problem ($k=n-1$) to chasing unions of two intervals on a line.
For any hypothesized upper bound $B$ on the competitive ratio of chasing unions of two intervals, choose $n$ such that $f(n) > 4B$,
and use the following reduction. Partition the interval $[0,1]$ into $2n-1$ equal-sized closed intervals $I_1,I_2,\dots,I_{2n-1}$, where
$I_j = \left[\frac{j-1}{2n-1},\frac{j}{2n-1}\right]$. The $n$ odd-indexed intervals correspond to the $n$ points $\{1,2,\dots,n\}$.
The position of the chaser indicates the location of the ``hole'' (i.e., the single point that is not occupied by a server). Clearly, the chaser
may choose to be placed in the interior of an even-indexed interval. In this case, we interpret the position of the hole as being in the
nearest odd-indexed interval, with the middle point associated (arbitrarily) to the right. A $k$-server request at point $i$ is mimicked by 
requesting $\left[0,\frac{2i-3}{2n-1}\right]\cup\left[\frac{2i}{2n-1},1\right]$. (Notice that if $i=1$ then the first interval is empty and if $i=n$ 
then the second interval is empty.) Thus, if the chaser's location interprets as point $i$, it must move a distance of at least $\frac{1}{4n-2}$ 
to chase the new request. In fact, to move to a location that would interpret as point $j$, it must move a distance of more than 
$\frac{|i-j|}{4n-2}$ and at most $\frac{2|i-j|}{2n-1}$. Therefore, the lower bound construction for the $k$-server problem, when translated
this way, forces a competitive ratio of at least $f(n)/4 > B$. This contradicts the hypothesis that $B$ is an upper bound on the competitive
ratio of chasing unions of two intervals. As any value $B\ge 1$ is contradicted, the competitive ratio is unbounded.
\end{proof}

We now proceed to the proof of the main lower bound.

\begin{proof}[Proof of Theorem \ref{thm:lb0}.]

We first assume $k=d=2$ and use the following gadget. For an interval $[a,b]$ let $[a,b]_i$ denote the line segment from $(a,i)$ to $(b,i)$.For $(a,b,c,d)$ we consider the quadrilateral $Q(a,b,c,d)$ with corners $(0,a),(0,b),(1,c),(1,d)$. For $a_1<a_2<b_1<b_2$ consider the following family $S(a_1,a_2,b_1,b_2)$ of $4$ sets:

\begin{enumerate}
    \item $\ell_0=[0,1]_0$
    \item $\ell_1=[0,1]_1$
    \item $Q(a_1,a_2,b_1,b_2)$
    \item $Q(b_1,b_2,a_1,a_2).$
\end{enumerate}

It is easy to see that the possible pairs $(s_1,s_2)$ of server positions which together meet all $4$ sets fall (up to ordering) into exactly two categories: 

\begin{enumerate}
    \item $s_1\in [a_1,a_2]_0$ and $s_2\in [a_1,a_2]_1$.
    \item $s_1\in [b_1,b_2]_0$ and $s_2\in [b_1,b_2]_1$.
\end{enumerate}

\begin{figure}[ht]
\caption{An example configuration $S(a_1,a_2,b_1,b_2)$ together with servers $s_1,s_2$ satisfying all $4$ requests $\ell_0,\ell_1, Q(a_1,a_2,b_1,b_2), Q(b_1,b_2,a_1,a_2).$ By using these sets, $2$-chasing convex bodies contains the problem of chasing the union of $2$ intervals $[a_1,a_2]\cup [b_1,b_2]\subseteq [0,1]$. This is impossible as shown in Lemma~\ref{lem:MTSlb}.  }
\centering
\includegraphics[width=1.0\textwidth]{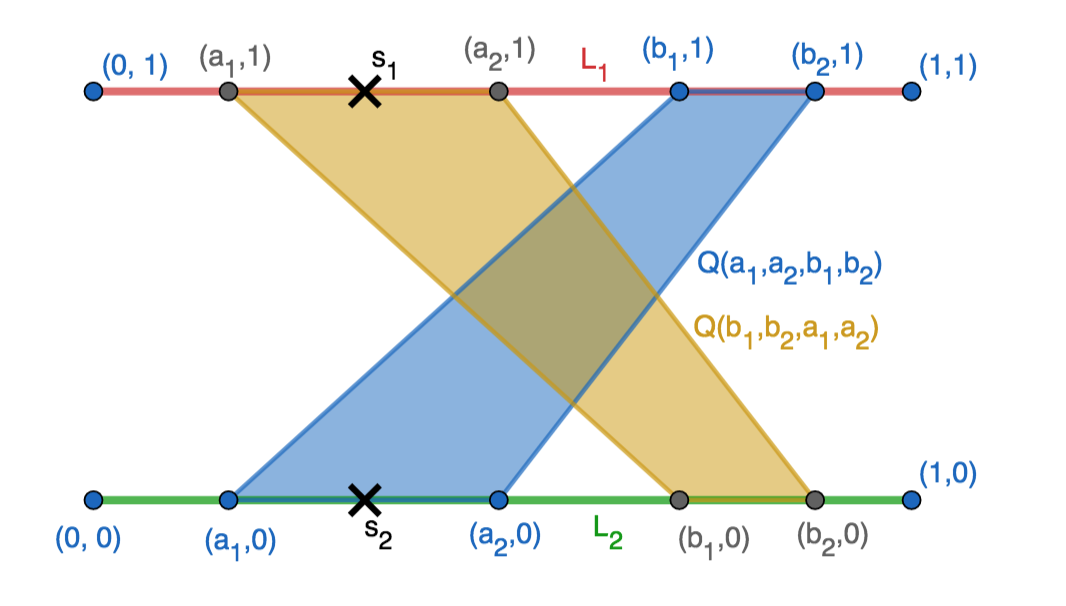}
\end{figure}

By repeating the $4$ sets $S(a_1,a_2,b_1,b_2)$ many times for each choice of $(a_1,a_2,b_1,b_2)$, we may assume that any algorithm moves into one of the two above types of configurations when given each new tuple $(a_1,a_2,b_1,b_2)$. We may then assume WLOG that $s_1$ and $s_2$ share $x$-coordinates throughout the entire process (e.g., by assuming that movement within a requested set is free for the algorithm). \\

Having made these reductions, the remaining problem is precisely chasing unions of two intervals in $[0,1]$. By Lemma~\ref{lem:MTSlb} this is impossible to do competitively, even with randomization. This proves the theorem in the case $k=d=2$.\\

Finally we explain how to reduce to the case $k=d=2$. For $d>2$, we may keep all requests in a $2$ dimensional subspace, which immediately reduces to $d=2$. For $k>2$, we may fix $k-2$ far away points and request each of them many times between any of the $S(a_1,a_2,b_1,b_2)$ requests. This forces any algorithm to incur huge cost unless it keeps a server at each of these far away points, reducing attention to $k=2$.

\end{proof}

\section{Beyond impossibility with well-sharpened functions}
In this section we go beyond the impossibility result of Theorem \ref{thm:lb0} by considering a special class of cost functions. We introduce the class of {\em well-sharpened} functions, which include in particular functions of the form $f_t(x)=|x-x_t|^{\gamma}$ for any $\gamma\geq 1$ 
(that is, it covers the $k$-median and $k$-means objectives), and we show how to $k$-chase competitively such 
functions in $\RR^d$, for any $k$ and $d\geq1$. In fact the results in this section apply more generally to arbitrary metric space. Thus for the rest of the section we fix a metric space $(X, \mathrm{dist})$.

In Section \ref{sec:kmedian} we start by looking specifically at the case of the $k$-median objective, that is the cost function $f_t$ is the form $f_t(x)=c_t\mathrm{dist}(x,x_t)$ for some $c_t \in (0,1)$ (we assume wlog $c_t <1$, since one can always repeat such cost functions if necessary), and we show in this case a reduction to $k$-server in $X$. Then in Section \ref{sec:wellsharp} we introduce our class of well-sharpened functions. Finally in Section \ref{sec:reductiontosharpness} we reduce well-sharpened functions to $k$-median. We remark that our algorithms have the additional property that they only move to the minimizer of a request $f_t$.

\subsection{Online $k$-median} \label{sec:kmedian}
We reduce online $k$-median to $k$-server by following a similar reduction as in \cite{antoniadis2016chasing} where the authors reduced chasing lazy convex bodies to chasing convex bodies:
 
\begin{thm}

\label{thm:kmedian}

Let $A_S$ be a randomized $k$-server algorithm which is $\alpha$-competitive against the oblivious (resp. adaptive online) adversary. Define an online $k$-median algorithm $A_F$ as follows. Upon arrival of a new request $f_t(x)=c_t\mathrm{dist}(x,z_t)$, $A_F$ will pass the request $z_t$ to $A_S$ with probability $c_t$, and ignore the request with probability $1-c_t$. Then $A_F$ is an $4\alpha$-competitive randomized algorithm for online $k$-median against the oblivious (resp. adaptive online) adversary. 

\end{thm}

Thus, a randomized algorithm for $k$-server that is $\mathrm{polylog}(k)$-competitive against the oblivious adversary 
(e.g., see the proposed approach by J. Lee~\cite{Lee18} for such a result) would imply a randomized algorithm for online 
$k$-median that is $\mathrm{polylog}(k)$-competitive against the oblivious adversary. Similarly the deterministic 
$(2k-1)$-competitive algorithm (work-function algorithm) for $k$-server~\cite{KP95} implies that the \emph{randomized work-function algorithm} for online 
$k$-median obtained from Theorem~\ref{thm:kmedian} is $O(k)$-competitive against the adaptive online adversary. Interestingly note that for $X=\RR^d$ these competitive 
ratios are dimension-free, in stark contrast with general convex functions where even for $k=1$ a dimension dependence 
is necessary. We also note that the randomized work-function algorithm is $O(k \cdot \mathrm{polylog}(k))$-competitive against an adaptive offline
for $k$-median through the main result of~\cite{BBKTW90}, which states that an algorithm's adaptive offline compettiive ratio is at most its adaptive online competitive ratio times the best oblivious competitive ratio of any algorithm. In particular this implies the existence of a deterministic algorithm with the same $O(k \cdot \mathrm{polylog}(k))$ competitive ratio.

\begin{proof}[Proof of Theorem \ref{thm:kmedian}.]

Let $I_F$ be the input sequence of pairs $(c_t,z_t)$ and $I_S$ the random subsequence of points $z_t$. Let $OPT_F$ and $OPT_S$ be offline optima for the two problems. We claim that:

\[\mathbb E[Cost(A_F)]\stackrel{(1)}{\leq} 2\cdot \mathbb E[Cost(A_S)]\stackrel{(2)}\leq 2\alpha\cdot  \mathbb E[Cost(OPT_S)]\stackrel{(3)}\leq 4\alpha\cdot Cost(OPT_F).\]

To prove $(1)$ we observe that if at any time $t$ the $k$-server algorithm would pay movement cost $d_t$ if $z_t\in I_S$ (where $d_t$ depends on the previous randomness), then $A_S$ will have expected cost $c_td$ while $A_F$ will have expected cost at most $2c_td$ for combined service and movement cost. (Here we use the fact that $A_F$ pays for only the closest server.) Inequality $(2)$ is by definition of $A_S$ being competitive.

To prove $(3)$, we note that $OPT_F$ can be turned into a solution to $I_S$ with expected cost at most $2\cdot Cost(OPT_F)$. Indeed, simply following the path of $OPT_F$ and moving the closest server to the random subsequence of requests and back achieves this. If $OPT_F$ has service cost $c_t\mathrm{dist}(x_{i,t},z_t)$ for this request, then the corresponding expected $k$-server cost is $2c_t\mathrm{dist}(x_{i,t},z_t)$.\\

The above argument was in the oblivious case. However essentially the same proof works for adaptive online adversaries. An adaptive online adversary for $k$-median is also an adaptive online adversary for the randomized instance of $k$-server, again with at most double the cost on average.

\end{proof}

\subsection{Well-sharpened functions and Move-to-Minimum Chasing} \label{sec:wellsharp}
We now introduce a new class of functions, for which we will be able to do a reduction to online $k$-median.

\begin{defn}

We say a function $f:X\to\mathbb R^+$ is $(\alpha,\beta)$-\emph{well-sharpened} (with $\beta<1<\alpha$) if $\alpha\mathrm{dist}(y,z)\geq \mathrm{dist}(x,z)$ implies that $\frac{\beta f(y)}{\mathrm{dist}(y,z)}\geq \frac{f(x)}{\mathrm{dist}(x,z)}$. 

\end{defn}

For instance, the cost function $f_t(x)=\mathrm{dist}(x,z_t)^{\gamma}$ is $(\alpha,\alpha^{\gamma-1})$ well-sharpened for any $\alpha >1$ in any metric space. In particular note that well-sharpened functions include some non-convex functions. Any $\kappa$ well-conditioned convex function is $(\alpha,\kappa\alpha)$ well-sharpened in any normed space. Moreover a positive linear combination of $(\alpha,\beta)$ well-sharpened functions is $(\alpha,\beta)$ well-sharpened.

For a $(\alpha,\beta)$ well-sharpened function $f$ around the point $z$, we define its $k$-median replacement $\tilde f^{x}$ relative to input point $x$ to be the function $\tilde f^{x}(y)=c \cdot \mathrm{dist}(y,z)$ where $c=\frac{f(x)}{\mathrm{dist}(x,z)}$ is chosen so that $f(x)=\tilde f^x(x)$. Crucially note that:
\begin{equation} \label{eq:forlater}
\tilde{f}^x(y) \geq \frac{1}{\alpha} \tilde{f}^x(x) \Leftrightarrow \alpha \mathrm{dist}(y,z) \geq \mathrm{dist}(x,z) \Rightarrow f(y) \geq \frac{f(x) \cdot \mathrm{dist}(y,z)}{\beta \mathrm{dist}(x,z)} = \frac{1}{\beta} \tilde{f}^x(y) \,.
\end{equation}

We define a restricted class of chasing algorithms which includes all $k$-server based $k$-median algorithms from Theorem~\ref{thm:kmedian}.

\begin{defn}

A \emph{move to minimum} (MTM) algorithm for chasing well-sharpened functions responds to any request by moving a server to the minimum of the request or not moving any servers.

\end{defn}

\subsection{From well-sharpened functions to $k$-median} \label{sec:reductiontosharpness}

We consider here a variant of the chasing problem which we term \emph{blind chasing}: the service cost $f_t$ is paid {\em before} the algorithm can move. Clearly this modification does not affect OPT, and a priori it makes the problem more difficult for ALG. In fact blind chasing is no more difficult than regular chasing for $1$-Lipschitz functions. Note that the remainder of this section will only consider adaptive adversaries.

\begin{lemma} \label{lem:blind1}

A $C$-competitive chasing algorithm against an adaptive sequence of requests which are $1$-Lipschitz is also $2C$-competitive for blind chasing against the same request sequence.

\end{lemma}

\begin{proof}

By definition the blind service cost of ALG is at most the original service cost plus the movement cost, as $f_t(x_{t})\leq f_t(x_{t+1})+ \mathrm{dist}(x_{t+1},x_{t})$ for any server position $x_t$ of ALG at time $t$. 

\end{proof}

\begin{lemma} \label{lem:blind2}

MTM chasing of $(\alpha,\beta)$ well-sharpened functions with competitive ratio $C$ can be reduced to MTM chasing of $(\alpha,\beta)$ well-sharpened functions with competitive ratio $C$ and the additional restriction that $f_t(x_{t,j})\leq dist(z_t,x_{t,j})$ for all ALG servers $x_{t,j}$.

\end{lemma}

\begin{proof}

Simply replace the request $f_t$ with $M_t$ requests $\frac{f_t}{M_t}$, where $M_t$ is sufficiently large. The MTM restriction means we can choose $M_t$ once for each $f_t$ (otherwise we would have to adjust based on the movements of ALG).

\end{proof}

We will now show how to reduce blind chasing of well-sharpened requests satisfying $f_t(x_{t,j})\leq dist(z_t,x_{t,j})$ (which is at least difficult as ordinary chasing of arbitrary functions thanks to Lemma~\ref{lem:blind2} and the fact that blindness only makes things more difficult) to blind chasing of $k$-median objectives $c_t|x-z_t|$ with $c_t\leq 1$ (which we can already do thanks to Theorem \ref{thm:kmedian} combined again with Lemma \ref{lem:blind1}). Importantly this reduction is inherently adaptive, so it requires an algorithm for online $k$-median which is competitive with an adaptive (online or offline) adversary, and in particular we do not obtain polylogarithmic type competitive ratio for well-sharpened functions.

\begin{lemma}

\label{lem:wellcentered}

Suppose $A$ is an online MTM algorithm for blind online $k$-median with $1$-Lipschitz costs which is $4\alpha$ competitive against the adaptive offline (resp. online) adversary. Define an algorithm $\tilde A$ for blind chasing of $(10\alpha,\beta)$-well-sharpened functions satisfying $f_t(x_{t,j})\leq dist(z_t,x_{t,j})$ as follows. If the $t$-th request $f_t$ is centered at $z_t$, and the closest online server to $z_t$ is $\hat x_{t}$, then pass the function request $\tilde f_t=\tilde f_t^{\hat x_{t}}(\cdot)$ to $A$ and move accordingly. Then $\tilde A$ is MTM and $O(\alpha\beta)$ competitive against the adaptive offline (resp. online) adversary.

\end{lemma}

\begin{proof}

Let $I_F$ be the sequence of true cost functions and $\tilde I_F$ the sequence of $k$-median functions. We denote by $Cost_A(f_t)$ the service cost from $f_t$, etc. We give the proof for the adaptive offline adversary; the adaptive online version is verbatim identical 
except for adding expectations around the cost of OPT.

The idea of the proof is that by Theorem~\ref{thm:kmedian}, only a small fraction of the cost of OPT against the sequence $\tilde I_F$ comes from service cost close to the center points $z_t$. When OPT does not have a server close to $z_t$, we can lower bound its service cost by the service cost from $\tilde f_t^{x_{t,i}}$. More precisely given \emph{any} OPT path, we define the set of \emph{bad times} $B\subseteq [T]$ via:

\[B=\bigg\{t\in [T]\bigg| Cost_{OPT}(\tilde f_t)\leq \frac{1}{10\alpha} Cost_{A}(\tilde f_t)\bigg\}.\]

We have:

\[\sum_{b\in B} Cost_{OPT}(\tilde f_b) \leq \frac{1}{10\alpha} Cost_A(\tilde I_F) \leq \frac{Cost_{OPT}(\tilde I_F)}{2}. \]

On the set $B^c=[T]\backslash B$, letting $y_{t}$ be the minimum cost OPT server at time $t$, we have that \eqref{eq:forlater} implies that  $f_t(y_t) \geq \beta^{-1}\tilde f_t(y_t)$. In particular we obtain

\begin{align*}Cost_{OPT}(I_F)&\geq \left(\sum_{t\in B^c} Cost_{OPT}(f_t)\right)+MovementCost(OPT) \\
&\geq \frac{1}{\beta}\left(\left(\sum_{t\in B^c} Cost_{OPT}(\tilde f_t)\right)+MovementCost(OPT\right)\\
&= \frac{1}{\beta}\left(Cost_{OPT}(\tilde I_F)-\sum_{b\in B} Cost_{OPT}(\tilde f_b)\right)\\
&\geq \frac{Cost_{OPT}(\tilde I_F)}{2\beta}.
\end{align*}

This means we have

\[Cost_{ALG}(I_F)=Cost_{ALG}(\tilde I_F) \leq O(\alpha)\cdot Cost_{OPT}(\tilde I_F) \leq O(\alpha\beta)\cdot Cost(OPT)(I_F)\]

 as desired.

\end{proof}

Let us call the result of applying Lemma~\ref{lem:wellcentered} to the randomized work-function algorithm constructed previously the \emph{well-sharpened randomized work-function algorithm} or WSRWFA. Then we obtain:

\begin{cor}

Cost functions of the form $f_t(x)=\mathrm{dist}(x,z_t)^{\gamma_t}$ with $\gamma_t\in [1,\Gamma]$ can be chased with competitive ratio
$O(k)^{\Gamma}$ against the adaptive online adversary in any metric space via the WSRWFA. Moreover WSRWFA has adaptive offline competitive ratio $O(k\cdot \mathrm{polylog}(k))^{\Gamma}$ for the same family of functions.

\end{cor}

\begin{cor}

$\kappa$ well-conditioned cost functions can be chased with competitive ratio 
$O(\kappa\cdot k^2)$ against the adaptive online adversary in any normed space via the WSRWFA. Moreover WSRWFA has adaptive offline competitive ratio $O(\kappa\cdot k^2\cdot \mathrm{polylog}(k))$ for the same family of functions.

\end{cor}

We emphasize that Lemma~\ref{lem:wellcentered} only applies for competitive ratios against adaptive adversaries and in particular does not give a polylogarithmic algorithm against the oblivious adversary. Establishing such a result for the oblivious adversary, say in the case of $k$-means, would be very interesting. On the other hand, using a deterministic algorithm for online $k$-median in Lemma~\ref{lem:wellcentered} would result in a deterministic algorithm for well-sharpened functions. 
Also, Theorem~\ref{thm:kmedian} and Lemma~\ref{lem:wellcentered}
do not explicitly construct deterministic algorithms but instead establish a competitive ratio against the adaptive offline adversary, which
is equal to the deterministic competitive ratio (see~\cite{BBKTW90}).

\section{Regret analysis}

In this section we consider the online learning version of the problem. Namely, at each time $t\in [T]$ the decision maker chooses $k$ points $x_{t,1},\dots,x_{t,k}$ in the unit Euclidean ball $B$. Then a $1$-Lipschitz convex function $f_t$ on the Euclidean unit ball is revealed, and the associated cost is the ``top-$k$'' value, that is $\min_{1 \leq j \leq k} f_t(x_{t,j})$. We are interested in controlling the regret defined as:
\[
\sum_{t=1}^T \min_{1 \leq j \leq k} f_t(x_{t,j}) - \min_{x^1, \hdots, x^k \in B} \sum_{t=1}^T \min_{1 \leq j \leq k} f_t(x^j) \,.
\] 
For $k=1$, it is well-known that one can obtain a computationally efficient and deterministic algorithm with regret $O(\sqrt{T})$. On the other hand for $k>1$ it is straightforward to obtain the following guarantee:

\begin{prop} \label{prop:naive}

There exists a computationally inefficient randomized algorithm for top-$k$ action OCO with regret $O(\sqrt{k d T \log(T)})$. In fact, we only require that the loss functions be Lipschitz, and could use a different loss for each action (where the choice of $k$ actions is viewed as an ordered list).

\end{prop}

\begin{proof}

We discretize the unit ball into $O(T)^d$ subsets of diameter at most $\frac{1}{T}$. Viewing the collection of $k$-tuples of subsets as $O(T)^{dk}$ distinct arms, we obtain the stated bound $O(\sqrt{T\log(O(T)^{dk})})=O(\sqrt{k d T \log(T)})$ using multiplicative weights.

\end{proof}

Compared to the case $k=1$, we note four differences: (i) computational inefficiency, (ii) randomization, (iii) dimension dependency, and (iv) extra $\sqrt{k}$ factor. In the following we prove that (i), (ii) and (iv) are actually unavoidable. Regarding (iii) we do not obtain a definitive answer, but we show that for the specific case of {\em linear} losses one can actually remove the dimension dependency. Finally we also prove that, again on the contrary to the case $k=1$, assuming strong convexity of the losses does not allow to reduce the regret below $\sqrt{T}$.

\subsection{Computational hardness}
We show here that for top-$k$ action online linear optimization (OLO) there does not exist a computationally efficient algorithm achieving sublinear regret. We give a reduction from top-$2$ action online \emph{linear} optimization to approximating the $L^2\to L^1$ norm of a matrix, which is known to be reducible to MAXCUT.

\begin{thm}

For some small constant $c$ it is NP-hard to obtain $cT/d$ regret (or even estimate the optimal expected per-timestep loss up to $c/d$ additive error) for stochastic online linear optimization using $k=2$ actions.

\end{thm}

\begin{proof}

We consider stochastic requests which are linear functions $f(x)=x\cdot w$ where $w$ is sampled from a known symmetric distribution over $d$-dimensional unit vectors. Without loss of generality we may assume that only pairs $(x,-x)$ are chosen for the action pairs. Indeed, it is easy to see from the symmetry that switching from $(x_1,x_2)$ to either $(x_1,-x_1)$ or $(x_2,-x_2)$ each with probability $1/2$ always yields an expected improvement. As we have a known distribution, the problem hence amounts to maximizing $\mathbb E[|x\cdot w|]$. \\

Now, we suppose the distribution of $w$ is given in the form of a $d\times d$ matrix $W$ of equiprobable rows, where for each $1\leq i\leq d$ there is a $\frac{1}{2d}$ chance to receive cost $f(x)=x\cdot w_i$ and $\frac{1}{2d}$ chance to receive $f(x)=-x\cdot w_i$. Then $\mathbb E[f(x)]=\frac{-|Wx|_{L^1}}{d}$. Hence our problem reduces to approximating the $L^2\to L^1$ norm of $W$. According to section 1.3.2 of \cite{matrixnorm} it is NP hard to approximate this norm up to relative error $c$ for some small constant $c$. (Roughly, the matrix $WW^T$ maps from $L^{\infty}\to L^1$ and computing its norm subsumes MAXCUT.)\\

Now, any nonzero matrix can be scaled to have its largest row a Euclidean unit vector so without loss of generality we assume $W$ is of this form (which is necessary for the corresponding cost functions to be $1$-Lipschitz). Then we have $|W|_{L^2\to L^1}\geq 1$ and so we obtain NP hardness of approximating its norm within additive error $c$. Returning to the original problem, this means is NP hard to approximate the value of the optimal point within additive error $\frac{c}{d}.$ Hence assuming $P\neq NP$ it is impossible to obtain regret $cT/d$ in time $poly(d)$.

\end{proof}

\subsection{Necessity of randomization}

We show here that there is no top-$k$ action analog of online gradient descent, in the sense that no deterministic algorithm can achieve sublinear regret even in dimension $1$.  \\

\begin{thm}

For $1$-Lipschitz and convex losses on $[0,1]$, every deterministic algorithm achieves regret $\Theta(T)$ in the worst case. Moreover this holds even if we fix any strictly increasing, continuously differentiable convex function $f:\mathbb R^+\to\mathbb R^+$ and restrict all loss functions to take the form $\ell(x)=f_y(x):=f(|x-y|)$.

\end{thm}

We first remark that it is trivial to extend this lower bound to higher dimensions. Indeed, if we fix some line segment in $\mathbb R^d$ and only play functions of the form $f(|x-y|)$ for $y$ on the segment, any algorithm's performance is improved by projecting onto the line segment.

Our proof uses the following strategy: for any deterministic algorithm we design an adaptive adversary who achieves loss at least $f\left(\frac{1}{2k}\right)$ per time-step. On the other hand, we show using an averaging argument that for any adversarial sequence of functions $f_y(\cdot)$, there is an algorithm for the player which achieves smaller loss per time-step. Hence the adaptive adversary we designed forces $\Theta(T)$ regret on any deterministic algorithm.

\begin{proof}

We first claim that for any fixed choice of $k$ actions $0\leq s_1\leq s_2\leq \dots\leq s_k\leq 1$, an adaptive adversary can force loss at least $f\left(\frac{1}{2k}\right)$ using functions $f_y(\cdot)$. Indeed by the pigeonhole principle one of the following holds:

\begin{enumerate}
    \item $s_1\geq \frac{1}{2k}$.
    \item $s_j-s_{j-1}\geq \frac{1}{k}$ for some $j\leq \{2,3,\dots,k\}$.
    \item $s_k\leq 1-\frac{1}{2k}$.
\end{enumerate}

In the first case we take $\ell(x)=f_0(x)$, in the second we take $\ell(x)=f_{\frac{s_j+s_{j-1}}{2}}(x)$ and in the third we take $\ell(x)=f_1(x)$. This proves the first claim.

Next we claim that against an arbitrary sequence of adversarial plays, there exists a static top-$k$ action configuration with loss per time-step less than $f\left(\frac{1}{2k}\right)$. We accomplish this in two parts. We do this by exhibiting a fixed random configuration (i.e. independent of the adversary) which achieves this guarantee in expectation.

As first step, we consider a random, small $\varepsilon$-shift of the evenly spaced configuration $s_j=\frac{2j-1}{2k}$, i.e. the configuration $s_j=\frac{2j-1}{2k}+\eta\varepsilon$ for $\eta\in\{\pm 1\}$ uniformly random. This achieves average loss $f\left(\frac{1}{2k}\right)-\Omega(\varepsilon)$ for functions $f_y(\cdot)$ with $y$ in the interval $I_k=\left[\frac{1}{2k},1-\frac{1}{2k}\right]$. Indeed, the distance from such a $y$ to the closest action is always at most $\frac{1}{2k}$ and has probability at least $\frac{1}{2}$ to be at most $\frac{1}{2}-\frac{\varepsilon}{10}$. Moreover because $f$ is twice differentiable this random shift achieves average loss $f\left(\frac{1}{2k}\right)+h(\varepsilon)$ on $f_y(\cdot)$ for $y$ in the complement $I_k^c$, where $h(\varepsilon)=o(\varepsilon)$ is some function depending on $f$. Here the worst case is $y=0$ or $y=1$, and we use the fact the continuous differentiability of $f$ at $\frac{1}{2k}$ implies $f(\frac{1}{2k}\pm \varepsilon) = f(\frac{1}{2k})\pm f'(\frac{1}{2k})\varepsilon+o(\varepsilon)$.

We next observe that when $(s_1,s_k)=\left(\frac{1}{4k},1-\frac{1}{4k}\right)$ the loss is strictly less than $f\left(\frac{1}{2k}\right)$ on $I_k^c$, and bounded on all of $[0,1]$. Therefore, mixing in any such a configuration to the above random $\varepsilon$-shift with $\Theta(\sqrt{\varepsilon\cdot h(\varepsilon)})$ probability yields a distribution over $k$ actions achieving loss $f\left(\frac{1}{2k}\right)-\Omega(\sqrt{\varepsilon\cdot h(\varepsilon)})$ against any function $f_y(\cdot)$.

By averaging, this implies that against any adversary playing only functions $f_y(\cdot)$, there is (in hind-sight) a static $k$ action configuration achieving loss at most $f\left(\frac{1}{2k}\right)-\delta$ per time-step for $\delta=\delta(f,k)>0$. Applying this to the adversary obtained in the first part of this proof proves the theorem.

\end{proof}

\subsection{Lower Bound for Many Actions}

Now we return to consideration of randomized algorithms. Previously we showed 
$O(\sqrt{kdT\log(T)})$ regret is achievable (by an inefficient algorithm). Here we 
give a lower bound $\Omega\left(k^{-\frac{2}{d-1}}\cdot\sqrt{kT}\right)$ for any 
$k\geq 2$, which applies even for $k$-means. In particular this shows that 
logarithmic regret against strongly convex functions, another desirable property 
of online gradient descent, is unachievable for multi-action optimization.

\begin{thm} \label{thm:manyservers}

Consider instances where the loss functions are of the form of the $k$-means objective. Any randomized 
algorithm has expected regret $\Omega\left(k^{-\frac{2}{d-1}}\cdot\sqrt{kT}\right)$ when $k\geq 2,d\geq 6$. 
Moreover, this lower bound holds against stochastic adversaries with known distribution.

\end{thm}

The idea is to play stochastically from a clustering problem with many distinct optimal clusterings. One of these optima will gain $\Omega(\sqrt{kT})$ in performance due to random fluctuation, and no algorithm can leverage this gain. The simplest example is as follows: suppose $k-2$ and $d=1$, and the adversary plays the loss functions $f_y(x)=|x-y|^2$ where $y\in\{0,\frac{1}{2},1\}$ is uniformly random. Then two optimal (in expectation) configurations for the player are $(s_1,s_2)=(0,\frac{3}{4})$ and $(s_1,s_2)=(\frac{1}{4},1).$ However, the $\sqrt{T}$-size random fluctations in frequency mean that regret $\Omega(\sqrt{T})$ is unavoidable. 

\begin{proof}

For a general even $k=2j$ we extend the construction just outlined as follows. We first suppose that $k\leq 2^d$ for convenience. Create $j$ distinct regions $R_1,\dots, R_j$ containing a radius $\frac{1}{10}$ ball on the unit sphere separated by distance at least $\frac{1}{10}$ each. Within each region, form an equally spaced triple of points $p^i_A,p^j_B,p^j_C$ in a row with spacing distance $1/100$. The adversary will play randomly on these $3j$ points.\\

First we claim that all optimal clusterings are given by choosing within each region $R_i$ two points:  $p^i_A,\frac{p^j_B+p^j_C}{2}$ or $p^i_C,\frac{p^j_B+p^j_A}{2}$. To see this, note that in one of these clusterings, the total cost in a cluster with $2$ points is $\frac{1}{2\cdot 10^4}$. On the other hand clusters containing $\{p^i_A,p^i_C\}$ and possibly $p^i_B$ have total cost at least $\frac{2}{10^4}$. Moreover any other hypothetical cluster of size $m\geq 2$ which involves points from multiple regions has total cost at least $\frac{m-1}{2\cdot 10^4}$, since any cluster center is within distance $\frac{1}{20}$ of at most $3$ points in its cluster. We see that if the cluster sizes are $m_1,\dots,m_k$ then the total cost is at least $\sum_i \frac{m_i-1}{2\cdot 10^8}=\frac{h-k}{2\cdot 10^4}$ and equality is attained only on the clusterings described above.\\

Now as a result of random fluctuations it is easy to see that paying $\Omega(\sqrt{T/k})$ regret per region is unavoidable which results in expected regret $\sqrt{kT}$ in total. If $k$ is odd we can introduce another region with a single point and easily reduce to the even case.\\

In the case that $k\geq 2^d$, we simiarly create regions with radius and separation distance at least $\frac{1}{10k^{\frac{1}{d-1}}}$. The analysis is essentially unchanged; we lose a factor $k^{-2/(d-1)}$ because the scale of the losses must shrink due to the lack of room.




\end{proof}

\subsection{Dimension-free top-$k$ action OLO}
Our last result is a positive one for online linear optimization. Recall that in the case $k=1$, online convex optimization is equivalent to online linear optimization (indeed there randomness does not help since one can always do better by averaging, and for a deterministic strategy the adversary is always better off by playing a linear function). This does not seem to be true with multiple actions; top-$k$ action OLO is a nontrivial but inequivalent special case of top-$k$ action OCO. 

In this section we prove that $\tilde{O}(\sqrt{k T})$ is achievable for top-$k$ action OLO, thus closing the gap with the lower bound given in Theorem \ref{thm:manyservers}. Whether such a bound holds true for OCO is left as an open problem, but we note in Section \ref{sec:obstruction} that the approach proposed here cannot work for general convex functions.

Our approach is easier to explain for {\em stochastic} OLO, that is when the linear functions $f_1,\hdots, f_T$ are i.i.d. from some unknown probability distribution over $B$. We focus on this case in Section \ref{sec:ERM}, and in Section \ref{sec:extension} we discuss how to extend it to the adversarial case.

\subsubsection{Empirical risk minimization} \label{sec:ERM}
We propose to study the follow the leader algorithm (also known as empirical risk minimization), i.e.,
\[
\{x_{1,t+1}, \hdots, x_{k,t+1}\} \in \argmin_{x^1, \hdots, x^k \in B} \frac{1}{t} \sum_{s=1}^t \min_{1 \leq j \leq k} f_s \cdot x^j \,.
\]
It is clear (and standard) that the instantaneous loss at round $t+1$ of this algorithm is upper bounded by twice the generalization error, namely:
\[
\sup_{x^1, \hdots, x^k \in B} \left( \frac{1}{t} \sum_{s=1}^t \min_{1 \leq j \leq k} f_s \cdot x^j - \E_{f} \left[\min_{1 \leq j \leq k} f \cdot x^j \right] \right) \,.
\]
Via the symmetrization trick, in expectation the above quantity is upper bounded by twice the Rademacher complexity, namely:
\begin{equation} \label{eq:rademacher}
\frac{1}{t} \E_{\epsilon_1, \hdots, \epsilon_t} \sup_{x^1, \hdots, x^k \in B} \sum_{s=1}^t \left( \epsilon_s \min_{1 \leq j \leq k} f_s \cdot x^j \right) \,,
\end{equation}
where $\epsilon_s$ are i.i.d. random signs. This quantity is {\em precisely} the one studied in \cite{Kon18} (see also \cite{foster2019}), where it is shown that it is of order $O\left( \sqrt{\frac{k \log(k) \log^3(t)}{t}} \right)$. Thus the above discussion proves the following result:

\begin{thm} \label{thm:stochasticOLO}
For top-$k$ action stochastic OLO, follow-the-leader has regret at most $O\left( \sqrt{kT \log(k) \log^3(T)}\right).$
\end{thm}

\subsubsection{Extension to the adversarial case} \label{sec:extension}
A long line of works, starting with \cite{rakhlin2010}, has shown how to extend the empirical risk minimization and Rademacher complexity arguments to the adversarial case. The crux of the matter is to control a variant of \eqref{eq:rademacher} known as the {\em sequential Rademacher complexity}. For example it is known that the classical Talgrand's contraction lemma also applies for sequential Rademacher complexity. In our case we need a vector-valued contraction lemma, as originally proved in \cite{maurer} and then refined in \cite{Kon18, foster2019}. It was communicated to us by Alexander Rakhlin \cite{Rak20} that in fact the calculations done in \cite{foster2019} can be repeated to obtain the same statement for sequential Rademacher complexity, which in particular allows to extend Theorem \ref{thm:stochasticOLO} to top-$k$ action adversarial OLO.

\subsubsection{A different approach is needed for convex functions} \label{sec:obstruction}
To generalize Theorem \ref{thm:stochasticOLO} to OCO one would need to control \eqref{eq:rademacher} for convex functions rather than linear functions. For $k=1$ we would be interested in the quantity
\[
\E_{\epsilon_1, \hdots, \epsilon_t} \sup_{x \in B} \sum_{s=1}^t \epsilon_s f_s(x) \,,
\]
where $f_1, \hdots f_t$ are $1$-Lipschitz convex functions. Unfortunately it is easy to see that in the worst case this quantity is $\tilde\Theta(\sqrt{dT})$. Indeed one can create $2^{\Omega(d)}$ disjoint regions on the boundary of the ball $B$, such that a Lipschitz function can take values $\{0,1\}$ on those regions in any combination possible. One can now simply draw $f_1, \hdots, f_t$ at random from this set of $2^{\Omega(d)}$ functions. As a result we do not know whether to expect a dimension dependence in the convex case and consider it a tantalizing open problem.

\bibliographystyle{plain}
\bibliography{biblio}

\appendix
\section{Proof of Theorem \ref{thm: one-dim competitive}} \label{app:A}

We adopt the potential function argument of~\cite{CKPV91} to this more general problem.
For two $k$-server configurations $X,Y$, let $\match(X,Y)$ denote the cost of the minimum 
cost matching between the points in $X$ and the points in $Y$. This is simply the following 
measure: sort both sets from left to right, $x_1,x_2,\dots,x_k$ and $y_1,y_2,\dots,y_k$, then 
set $\match(X,Y) = \sum_{i=1}^k \|x_i - y_i\|$. In particular, moving from configuration $X$
to configuration $Y$ costs $\match(X,Y)$. Also, for a $k$-server configuration $X$
and a request $f$, denote by $\serve(X,f)$ the cost of serving $f$ at the configuration $X$.
We define a potential function
$$
\Phi(X,Y) = 2k\cdot\match(X,Y) + 2\sum_{x,x'\in X} \|x - x'\|,
$$
where $X$ is a configuration of the algorithm and $Y$ is a configuration of the adversary
(an optimal offline strategy).

Consider at first a request $f$ that causes the algorithm to move just one server. Suppose
that in response, the adversary moves from $Y$ to $Y'$. This affects only the first term of
$\Phi$, and $\Phi$ increases by at most $2k\cdot \match(Y,Y')$. Now, consider the algorithm's 
move. We split the move into two parts: a move towards an adversary's server, and a move
past the last adversary's server in the direction of the move. Of course, either one of these
parts could be vacuous. Let $X,X',X''$ denote, respectively, the starting configuration, the
configuration at the end of the first part, and the configuration at the end of the move, respectively.
Thus, during the first part, the first term of $\Phi$ decreases by $2k\cdot \match(X,X')$ and
the second term of $\Phi$ increases by $2(k-1)\cdot \match(X,X')$. Thus, $\Phi$ decreases
by $2\cdot\match(X,X')$. If $X'' = X'$ (no second part), then either $\serve(X',f) = \serve(Y',f)$,
or $\serve(X',f) = \match(X,X')$. Thus, either way,
$$
\Phi(X',Y') - \Phi(X,Y)\le 2k\cdot (\match(Y,Y') + \serve(Y',f)) - (\match(X,X') + \serve(X',f)).
$$
If, on the other hand, $X''\ne X'$, then clearly $\match(X,X'')\le\serve(X'',f)\le\serve(Y',f)$.
The potential function $\Phi$ increases by $(4k-2)\cdot\match(X',X'')\le (4k-2)\cdot\match(X,X'')$.
So, 
\begin{eqnarray*}
\Phi(X'',Y') - \Phi(X,Y) & \le & 2k\cdot\match(Y,Y') + (4k-2)\cdot\match(X,X'') \\
& \le & 2k\cdot\match(Y,Y') + (4k-2)\cdot\serve(Y',f) \\
& \le & 4k\cdot(\match(Y,Y') + \serve(Y',f)) - 2\cdot\serve(Y',f) \\
& \le & 4k\cdot(\match(Y,Y') + \serve(Y',f)) - (\match(X,X'') + \serve(X'',f)).
\end{eqnarray*}

Now consider a request $f$ that causes the algorithm to move two servers. The effect of the
adversary's move is the same---the potential function $\Phi$ increases by at most
$2k\cdot \match(Y,Y')$. When the algorithm moves, we distinguish between the case that
there is an adversary server between the two moving servers of the algorithm, and the case
that no adversary server occupies this segment. Clearly, as the algorithm's servers move,
we may switch from the first case to the second case. So let $X,X',X''$, respectively, denote 
the configuration before the move, at the end of the first part of the move (when the first case
holds), and at the end of the entire move, respectively. In the first part of the move, it must be
the case that at least one of the two servers is moving towards its matched adversary server,
so $\match(X',Y')\le\match(X,Y')$. On the other hand, the distance between the two servers
shrinks by the movement cost, and for any other server, the distance to one moving server 
shrinks by the same amount as the distance to the other moving server grows.  Therefore,
the second term of $\Phi$ decreases by $2\cdot\match(X,X')$. If $X'' = X'$, we use again the
fact that either $\serve(X',f) = \serve(Y',f)$, or $\serve(X',f) = \match(X,X')$.
We get, as before,
$$
\Phi(X',Y') - \Phi(X,Y)\le 2k\cdot (\match(Y,Y') + \serve(Y',f)) - (\match(X,X') + \serve(X',f)).
$$
If $X''\ne X'$, then during the move from $X'$ to $X''$ both servers might be moving away 
from their matched adversary servers, so $\match(X'',Y') - \match(X',Y')\le \match(X',X'')\le \match(X,X'')$. 
The second term of $\Phi$ still decreases by $2\cdot\match(X',X'')$.
Thus, in total, $\Phi$ increases due to the algorithm's move by at most
$(2k-2)\cdot\match(X,X'')$. We use, as before, the fact that in this case $\match(X,X'')\le\serve(X'',f)\le\serve(Y',f)$.
We get
\begin{eqnarray*}
\Phi(X'',Y') - \Phi(X,Y) & \le & 2k\cdot\match(Y,Y') + (2k-2)\cdot\match(X,X'') \\
& \le & 2k\cdot\match(Y,Y') + (2k-2)\cdot\serve(Y',f) \\
& \le & 2k\cdot(\match(Y,Y') + \serve(Y',f)) - 2\cdot\serve(Y',f) \\
& \le & 2k\cdot(\match(Y,Y') + \serve(Y',f)) - (\match(X,X'') + \serve(X'',f)).
\end{eqnarray*}
As $\Phi(X,Y)\ge 0$ for all $X,Y$, by summing the change in $\Phi$ over all steps we conclude
that the algorithm's total cost is at $4k$ times the optimal cost, plus an additive term
equal to the initial value of $\Phi$ (which depends only on the initial configuration of the servers).

\section{Multi-armed Bandit with $k$ actions}

We consider the \emph{best of $k$ actions} online learning problem, where the player picks $k$ action and receives the minimum loss among them. We consider the full-feedback model, so all losses are observed. As usual the player aims to achieve small regret with respect to the best static $k$-subset. This is a simpler version of the top-$k$ online convex optimization we consider in the main body. Here we point out that the discretized upper bound in Proposition~\ref{prop:naive} cannot be improved without using the geometric structure of $\mathbb R^d$.

\begin{prop}

The minimax regret for best of $k$ actions online learning is $\Theta\left(\sqrt{k T \log\left(\frac{n}{k}\right)}\right).$

\end{prop}

\begin{proof}

The upper bound follows from viewing the problem as full-feedback online learning on all $\binom{n}{k}$ sets of $k$ actions. For the lower bound, we make each action have loss $1$ with probability $1-\frac{1}{k}$ and else $0$, independently over arms and time. Then any algorithm achieves loss $(1-\frac{1}{k})^k T$. However in hind-sight we can group the arms into batches of $\frac{n}{k}$. Then we consider the $k$-subset with $j$th action the arm in batch $j$ with best performance on the subset of times when the first $j-1$ actions all saw a loss. This gains $\sqrt{\frac{Tk\log(n/k)}{n}}$ per time-step via consideration of the maximum of $n/k$ standard Gaussians (which is valid for large $T$ by the central limit theorem), completing the proof.

\end{proof}

\end{document}